\documentclass[11pt]{article} 
\usepackage[numbers]{natbib}
\usepackage{booktabs} 
\usepackage[utf8]{inputenc} 

\usepackage{geometry} 
\usepackage{graphicx} 

\usepackage{booktabs} 
\usepackage{array} 
\usepackage{paralist} 
\usepackage{verbatim} 
\usepackage{subfig} 

\usepackage{fancyhdr} 
\usepackage{sectsty}
\allsectionsfont{\sffamily\mdseries\upshape} 

\usepackage[nottoc,notlof,notlot]{tocbibind} 
\usepackage[titles,subfigure]{tocloft} 

\usepackage[linesnumbered,noend,ruled,noline]{algorithm2e} 

\SetAlFnt{\small}
\SetAlCapFnt{\small}
\SetAlCapNameFnt{\small}
\SetAlCapHSkip{0pt}
\IncMargin{-\parindent}

\usepackage{amsmath}
\usepackage{amssymb}
\usepackage{amsthm} 
\usepackage{url}

\usepackage{color}
\usepackage{bm}
\usepackage{tikz}
\usepackage{thmtools, thm-restate}

\newtheorem{thm}{Theorem}
\newtheorem{theorem}[thm]{Theorem}

\newtheorem*{lemma*}{Lemma}

\usepackage{authblk}

\newcommand{\agents}{\mathcal{N}}
\newcommand{\items}{\mathcal{M}}
\newcommand{\envyfree}{\mathcal{E}}
\newcommand{\lotteries}{\mathcal{L}}

\newenvironment{talign*}
 {\csname align*\endcsname}
 {\endalign}

\usepackage{tikz}
\usetikzlibrary{arrows,petri,topaths}

\usepackage{float}

\begin{document}
\title{Beyond Cake Cutting: Allocating Homogeneous Divisible Goods\thanks{The second author was supported in part by an NSF CAREER award CCF-2047907 and NSF grant CCF-2008280. The third author was supported in part by a Google Research Scholar Award.
The fourth author was supported in part by NSF grants CCF-2008280 and CCF-1755955.  Part of this work was done while the last three authors were visiting the Simons Institute for the Theory of Computing.}}

\author[a]{Ioannis Caragiannis}
\author[b]{Vasilis Gkatzelis}
\author[c]{Alexandros Psomas} 
\author[b]{Daniel Schoepflin}

\affil[a]{Department of Computer Science, Aarhus University. Email: \texttt{iannis@cs.au.dk}}
\affil[b]{Department of Computer Science, Drexel University. Email: \texttt{\{gkatz,schoep\}@drexel.edu}}
\affil[c]{Department of Computer Science, Purdue University. Email: \texttt{apsomas@cs.purdue.edu}}

\date{} 
\renewcommand\Affilfont{\itshape\small}
\renewcommand\Authands{, }

\maketitle

\begin{abstract}
The problem of fair division known as ``cake cutting'' has been the focus of multiple papers spanning several decades. The most prominent problem in this line of work has been to bound the query complexity of computing an envy-free outcome in the Robertson-Webb query model. However, the root of this problem's complexity is somewhat artificial: the agents' values are assumed to be additive across different pieces of the ``cake'' but \emph{infinitely complicated} within each piece. This is unrealistic in most of the motivating examples, where the cake represents a \emph{finite} collection of homogeneous goods. 

We address this issue by introducing a fair division model that more accurately captures these applications: the value that an agent gains from a given good depends only on the amount of the good they receive, yet it can be an arbitrary function of this amount, allowing the agents to express preferences that go beyond  standard cake cutting. In this model, we study the query complexity of computing allocations that are not just envy-free, but also approximately Pareto optimal among all envy-free allocations. Using a novel flow-based approach, we show that we can encode the ex-post feasibility of randomized allocations via a polynomial number of constraints, which reduces our problem to solving a linear program.
\end{abstract}

\newpage

\section{Introduction}
The ``cake-cutting'' problem has played a central role in the history of fair division. The commonly used story behind this problem involves a heterogeneous cake with a variety of different toppings (e.g., a region with sprinkles, a region with more icing, or a cherry sitting at the top), and the goal is to fairly divide this cake among a group of children with different preferences. In reality, this is just a metaphor aiming to capture complicated resource allocation settings, where each region of the cake may correspond to a different resource and each agent competing for these resources may have a diverse and complicated valuation function.

The cake is mathematically represented using the interval $[0,1]$ and a piece of the cake corresponds to a (possibly not contiguous) subset of this interval. The agents' preferences are encoded by a valuation function that maps each sub-interval to a number, and the value of an agent for a set of disjoint intervals is equal to the sum over all these intervals of her value for each interval. So, the agents' values for each sub-interval can be highly complicated, but their value across these intervals is additive.

An important goal in cake-cutting is to find some assignment of pieces to the agents that is ``envy-free'', i.e., each agent (weakly) prefers her piece over the piece allocated to any other agent. The main obstacle in achieving this goal is the fact that, due to the high complexity of the agents' preferences, they cannot be succinctly reported to the algorithm. In the classic Robertson-Webb query model~\cite{RW1998}, access to the agents' valuations is provided only through two types of queries: i) \emph{value queries}, that take as input some interval and return the value of the agent for it, and ii) \emph{cut queries}, that take as input a point $\ell\in [0,1]$ and a value $v$, and return the leftmost point $r\in [\ell, 1]$ such that the value of the agent for the interval $[\ell, r]$ is equal to $v$.

A trivial way to achieve envy-freeness is to throw away all of the cake and thus ensure that no agent has anything to be envious about. 
As it turns out, even the simple restriction that none of the cake can be thrown away is enough to really complicate things. For instances involving up to three agents, this was solved in the 1960s \cite{BT1996} but, up until just a few years ago, we did not even know whether a \emph{bounded} (let alone polynomial) number of queries would be sufficient for instances involving four agents. This was recently settled in the affirmative by Aziz and Mackenzie~\cite{AM2016a,AM2016b}, but the proposed protocols are very complicated, and the required number of queries is superexponential in the number of agents.

Despite the extensive amount of attention that this problem has garnered, its apparent intractability is partly due to the artificial complexity of the model. 
For example, in many of the motivating application domains, the cake actually comprises a set of distinct homogeneous resources and each agent cares only about \emph{how much} of each resource she receives rather than \emph{which part} of it. In other words, when allocating a piece of the proverbial cherry on top of the cake to an agent, she may care about the size of this piece, but not worry about whether it came from the front or the back side of the cherry (it is all just cherry anyway!). But, introducing this type of structure into the cake-cutting model reduces it to the very special case of \emph{piecewise constant valuations}: the $[0, 1]$ interval is partitioned into sub-intervals, each of which corresponds to one of the resources, and each agent's value is constant within each sub-interval. This captures a limited range of agent preferences and even computing a welfare-maximizing envy-free outcome becomes trivial~\cite{CLPP2011}, since the agents can just directly report their (constant) value for each resource. The complexity of the cake-cutting model lies in the assumption that there is no such structure in general, even if the cake is partitioned into arbitrarily small pieces.

In this paper, we propose an alternative model that more accurately captures these applications, while also allowing natural types of agent preferences that are beyond the scope of the cake cutting model (e.g., complementarity and substitutability).
Achieving envy-freeness in this model is easy even without discarding any of the resources, just like it is in most motivating examples; one can just divide each resource equally among the agents. We instead aim to compute not just any envy-free allocation, but rather one that is (approximately) Pareto optimal among all envy-free allocations, which is not guaranteed if we just divide each resource equally. 

\subsection{Results and Techniques}
We study the problem of fairly allocating a set of divisible homogeneous resources among a set of agents. Deviating from the cake-cutting literature, we use a model where each agent's value depends only on \emph{how much} of each resource they receive rather than \emph{which part}. We allow each agent's value to be a complicated function of the amount of each resource that they receive, and we use cut and value queries, analogous to the Robertson-Webb query model, to ``discover'' the structure of these functions.

This model provides a better abstraction for many real-world applications. For example, the resource could represent the shares of a company where the value of an agent may increase dramatically upon receiving more than half of all the shares. Alternatively, if the resource being distributed is energy, then there may be a limit on how much energy an agent can utilize each day, so being allocated more than that would not offer her any additional value.
These are utilities that the standard cake cutting model cannot capture.

In Section~\ref{sec:beyond}, we explain why, in contrast to the cake-cutting model, the use of randomization in our model can enrich the space of achievable (expected) utility vectors, motivating the use of lotteries over deterministic outcomes. 
We focus on ``fair'' lotteries that satisfy \emph{ex-ante} envy-freeness (i.e., envy-freeness with respect to expected utilities) and our main goal is to compute one that is (approximately) Pareto efficient among them. Rather than aiming for \emph{ex-post} Pareto efficiency (which compares the outcome of the lottery to all \emph{deterministic} envy-free outcomes) we instead aim for the stronger guarantee of \emph{ex-ante} Pareto efficiency (which compares the lottery to the richer space of all envy-free lotteries).

In Section~\ref{sec:lowerbound}, we prove that finding an ex-ante envy-free lottery that is $\epsilon$-approximate ex-ante Pareto efficient among all ex-ante envy-free lotteries would require $\Omega(1/\epsilon)$ queries even for instances with just two agents and two resources. 

In Section~\ref{sec:upperbound}, we complement this negative result by providing an algorithm that finds such a lottery using only $O(1/\epsilon^2)$ queries for any constant number of agents and resources. We generalize this to the case of $n$ agents and $m$ resources using $O(nm/\epsilon^2)$ queries. The algorithm begins by issuing a sequence of value queries for each agent-resource pair, providing a discretized approximation of the true valuations. It then computes a lottery by solving a linear program using these approximate valuations. 
The main technical challenge is to ensure that the lottery computed by this linear program is \emph{ex-post feasible}, i.e., that the randomized allocation we obtain can be decomposed into a distribution over feasible deterministic outcomes. We achieve this through a flow-based technique which may be of independent interest. 

Finally, in Section~\ref{sec:future}, we conclude with some observations and open problems for future research within the new model.

\subsection{Related Work}
\paragraph{Cake-cutting.} 
A procedure for computing envy-free outcomes with a finite (but potentially unbounded) number of queries in the cake cutting model was found by \citet{BT1995}. \citet{KLP2013} gave a bounded protocol for any number of agents with piecewise linear valuations. Even for just four agents with general valuations, it was unknown if a protocol producing an envy-free allocation using a bounded number of queries exists, until the result of \citet{AM2016a}.  They then extended their own work to achieve a  procedure for $n$ agents using a bounded number of queries, but the number is a  tower of exponents of $n$~\cite{AM2016b}. \citet{AFMPV2018} improved the number of queries for four agents, but a further improvement on the general $n$-agent case has remained elusive, so there is still a massive gap between the upper bound and the lower bound of $\Omega(n^2)$ queries found by \citet{P2009}. 

\paragraph{Computing efficient envy-free allocations.} 
The type of guarantee that we seek in this paper, i.e., to find an envy-free solution that is efficient among other envy-free ones, has also been examined in the context of the standard cake-cutting model.
\citet{RP1998} demonstrated that if the entire valuation function of each agent is reported to a mechanism, then an approximately Pareto efficient allocation among all envy-free allocations can be found by approximating the valuation functions by piecewise constant valuation functions. 
\citet{CLPP2011} extended this result by providing an efficient algorithm to compute an approximately welfare-maximizing allocation among all envy-free allocations. 
Our work can then be viewed as an analogue to \cite{CLPP2011} in our model.  Whereas their result applies in the standard cake-cutting model and thereby can avoid the use of randomization, ours leverages randomness to apply to more general valuation functions than can be captured by cake-cutting.  
\citet{bei2012optimal} studied a similar objective of finding welfare-maximizing connected allocations but, instead, among the set of proportional allocations in the standard cake-cutting model.\footnote{In standard cake-cutting, any envy-free complete allocation is always proportional, unlike in our model as we demonstrate in Section \ref{sec:beyond}.} 

Notably, fairness constraints may decrease efficiency.   
\citet{brams2012maxsum} examined the standard cake-cutting model and identified sufficient conditions for a welfare-maximizing envy-free allocation to be Pareto efficient among the set of \emph{all} (not necessarily envy-free) allocations.  They also provided examples where no welfare-maximizing envy-free allocation is Pareto efficient among the set of all allocations although the loss in efficiency they show is small.  
In a similar vein, \citet{CKKK2012} quantified the extent to which various fairness notions may affect efficiency and demonstrated that enforcing envy-freeness as a hard constraint can greatly decrease the maximum obtainable welfare. 
\citet{brams2012maxsum} also pointed out the question of finding an efficient algorithm to compute welfare-maximizing envy-free allocations which are Pareto efficient with respect to \emph{all} allocations in the standard cake cutting model as an important and challenging open problem even in the simple case of piecewise constant valuations.  In Section \ref{sec:future}, we pose a related open question in our model which we believe is of comparable difficulty and interest.

\paragraph{Homogeneous goods.}  While the query model we examine in this paper is directly related to the query model of \citet{RW1998}, our valuation function model is most similar to the work of \citet{FT2014}.  In that paper, the authors examine the division of a single homogeneous resource where the agents' utilities are non-linear.  However, they examine an alternative notion of an agent's ``fair share'' (based on the average utility she would receive if all other agents had valuations identical to hers) and agents report allocations rather than reply to queries regarding their valuation functions.  \citet{BGR2020} similarly examine the problem of allocating a single homogeneous good but instead assume that the available amount of the resource is unknown but drawn from some known prior distribution and the valuation functions of the agents are fully known to the designer.  They then study the computational complexity of finding welfare-maximizing allocations subject to ex-ante envy-freeness (where randomness is with respect to the amount of good) and characterize the loss in efficiency by imposing envy-freeness as a constraint.

\paragraph{Lotteries and ex-post feasibility.}  The key technical challenge we face in finding ex-ante Pareto efficient solutions which can be derandomized into a distribution over feasible deterministic outcomes has been explored in many contexts in fair division.  In the case of one-sided matching, \citet{BM2001} gave the well-known ``probabilistic serial'' algorithm which directly computes expected assignments and utilizes the Birkhoff-von Neumann theorem \cite{bvn1,bvn2} to decompose the expected assignments into a lottery over feasible outcomes.  The Birkhoff-von Neumann theorem was further generalized by \citet{BCKM2013} beyond one-to-one settings.  As we discuss in Section \ref{sec:upperbound}, the methods of \cite{BCKM2013} cannot be adapted to our model, so we view our flow-based approach as complementary to these prior results.

\section{Preliminaries}
We consider a setting with a set $\agents$ of $n$ agents and a set $\items$ of $m$ divisible items.
Throughout the paper, we use the terms items, goods, and resources interchangeably. An \emph{outcome} $x$ defines for each agent $i\in \agents$ and item $k\in \items$ the fraction $x_{ik} \in [0,1]$ of item $k$ that is allocated to agent $i$. An outcome is \emph{feasible} if no item is over-allocated, i.e., $\sum_{i \in \agents}{x_{ik}} \leq 1$ for all $k \in \items$.  A \emph{lottery} $L$ is a probability distribution over outcomes and it is (ex-post) \emph{feasible} if it assigns positive probability only to feasible outcomes.

For each agent $i$, her preferences are defined through a set of non-decreasing and Lipschitz\footnote{A function $f$ is Lipschitz if $|f(x) - f(y)| \leq C |x-y|$ $\forall x,y$ and some constant $C$.}
valuation functions $\{f_{ik}\}_{k \in \items}$. Her utility for receiving a fraction $x_{ik}$ of item $k$ is $f_{ik}(x_{ik})$ and her utility for an outcome $x$ is
$u_i(x) = \sum_{k \in \items}{f_{ik}(x_{ik})}.$  In other words, the utilities are weakly monotone within items and additive across items. 
The (expected) utility of agent $i$ for a lottery $L$ which assigns probability $p_\ell$ to each outcome $x^\ell$ is
\begin{align*}
u_i(L) = \sum_{\ell}{p_\ell \sum_{k \in \items}f_{ik}(x^\ell_{ik})}.
\end{align*}

 Let $u_i(x_{j})$ denote the utility agent $i$ would receive in outcome $x$ if she were given the allocation of agent $j$.
We then say that an outcome $x$ is \emph{envy-free} if for all agents $i,j$ we have $u_i(x_i) \geq u_i(x_j)$. Similarly, let $u_i(L_{j})$ denote the utility $i$ would receive from lottery $L$ if she were given the randomized allocation of agent $j$. A lottery $L$ is \emph{ex-post envy-free} if it  assigns positive probability only to outcomes that are envy-free and it is \emph{ex-ante envy-free} if it is envy-free with respect to the agents' expected utilities, i.e., for all agents $i,j$ we have $u_i(L_i) \geq u_i(L_j)$.

An outcome $x$ is \emph{$\epsilon$-Pareto optimal} with respect to a set of outcomes $\mathcal{X}$ if there exists no alternative outcome $y \in \mathcal{X}$ such that $u_i(y) \geq (1+\epsilon) \cdot u_i(x)$ for all $i \in \agents$ and at least one of these inequalities is strict \cite{ruhe1990varepsilon,IWM17,friedman2019fair,ZP20}. 
We say that a lottery $L$ is \emph{ex-post $\epsilon$-Pareto optimal} with respect to a set of outcomes $\mathcal{X}$ if it assigns positive probability only to outcomes that are $\epsilon$-Pareto optimal with respect to $\mathcal{X}$. A lottery $L$ is \emph{ex-ante $\epsilon$-Pareto optimal} with respect to a set of lotteries $\lotteries$ if there is no alternative lottery $L' \in \lotteries$ 
with $u_i(L') \geq (1+\epsilon) u_i(L)$ for all $i \in \agents$ and at least one strict inequality.  For $\epsilon=0$ we retrieve exact Pareto optimality.  Our main result focuses on the class of ex-ante envy-free lotteries $\envyfree$ and, in accordance with Cohler \textit{et al.} \cite{CLPP2011}, we say a lottery $L$ is (approximately) \emph{ex-ante Pareto optimal EF} if it is in $\envyfree$ and also (approximately) ex-ante Pareto optimal with respect to all lotteries in $\envyfree$.

Since the $f_{ik}$ valuation functions can be highly complicated and not succinctly representable, in line with the Robertson-Webb cake-cutting model, we assume that the algorithm can gather parts of this information using two types of queries
-- \emph{value} and \emph{cut} queries. A value query $\textsc{Value}(f, z)$ takes in a valuation function and a $z$ in $[0, 1]$ and returns the value of $f(z)$, and a cut query \textsc{Cut}$(f, v)$ returns the minimum $z$ in $[0, 1]$ such that $f(z)=v$.  Our main algorithm (see Section \ref{sec:upperbound}) shows that we can actually compute approximately ex-ante Pareto optimal EF lotteries using only value queries. In Section~\ref{sec:future} we also demonstrate that cut queries are useful for computing lotteries in $\envyfree$ that are \emph{exactly ex-post} Pareto optimal with respect to any feasible outcome.  
Moreover, our lower bound on the query complexity of finding approximately ex-ante Pareto optimal EF lotteries in Section \ref{sec:lowerbound} applies to any algorithm that can use both cut and value queries.

\section{Moving Beyond Cake Cutting}\label{sec:beyond}

As we discussed in the introduction, when allocating a set of homogeneous goods (where agents care only about the amount of each good they receive rather than what specific part of the good), the preferences in the cake cutting model are reduced to piecewise constant valuations. This is due to the restrictive assumption that the value of an agent for any two pieces is equal to the sum of her value for each of the pieces. Therefore, the agent's value for half a cherry needs to be exactly half of her value for all of the cherry (otherwise combining two halves would not add up to the whole value). In our model, this implies that every agent's valuation function for a good needs to be linear, which is unnecessarily restrictive. 

Our model moves beyond this restriction, allowing us to encode complementarity or substitutability in the agents' preferences. For example, if an agent's valuation function for an item is highly convex, like ``type A'' in Figure \ref{fig:agentValues}, then the agent does not receive much value from that item unless she gets a lot of it. In fact, using an extreme convex function allows us to essentially capture the problem of allocating indivisible items: although the item can still be divided, if agents do not receive any value from it unless they get the whole item then dividing it would amount to discarding it. Thus, ex-post envy-freeness can be, in a sense, too stringent in our model, so similarly to the fair division literature for indivisible items, we use randomization to overcome this issue and achieve ex-ante fairness and efficiency guarantees.

The lack of linearity within an item also means that, unlike in the cake-cutting model, in our model complete allocations which are envy-free are not guaranteed to be proportional.  For instance, if two agents have highly convex valuation functions for a single item (e.g., two agents of ``type A'') then splitting the single item equally between them is envy-free, but not proportional\footnote{Proportionality requires that every agent receives a utility at least $1/n$ times her utility for being allocated everything in full.}: by convexity, her value for a $1/2$ fraction of the item is less than a $1/2$ fraction of her value for receiving the whole item. 

\begin{figure}[h]
\centering
\includegraphics[width=0.7\columnwidth,trim=0.1in 4in 0.1in 4in]{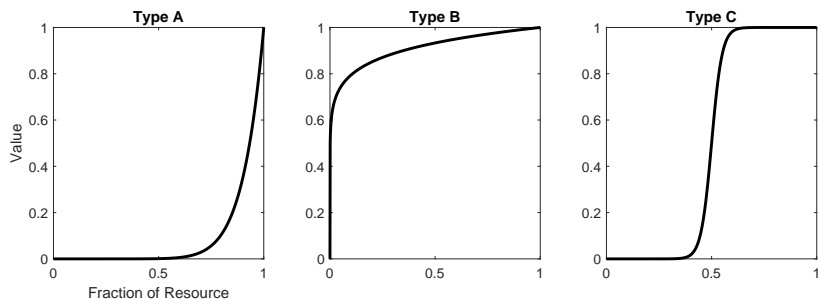}
\caption{Examples of agent valuation functions for an item.}
\label{fig:agentValues}
\end{figure}

\paragraph{Efficiency with non-linear valuation functions.}  
When the valuation functions are linear, as in the cake-cutting model, then agents are indifferent between a lottery and a (deterministic) outcome as long as their \emph{expected} allocation of each good in the former is the same as their deterministic allocation in the latter. As a result, there is no reason to use randomization. Moving beyond linear valuation functions, however, introduces several non-trivial opportunities to leverage randomization in order to make the agents happier. 
On one hand, agents with convex (``type A'') valuation functions are risk-seeking and would strictly prefer a lottery that gives them the full item with probability $1/2$ and nothing with probability $1/2$ over a deterministic outcome that always gives them $1/2$ of the item. On the other hand, agents with concave (``type B'') valuation functions are risk-averse and would strictly prefer the deterministic outcome.

Even for an instance involving two agents and a single item, how should we allocate the item?
Should we use a lottery? If so, what outcomes should the lottery randomize over, and how large does the support need to be in order to achieve ex-ante Pareto efficiency among all envy-free lotteries?
These are non-trivial questions that we urge the reader to contemplate before proceeding.

\paragraph{Outcomes in the utility space. } Although thinking about Pareto optimal lotteries in the allocation space can be quite confusing, visualizing them in the utility space is very convenient.  Each outcome defines a utility for each agent, and this point lies in some $n$-dimensional space \emph{regardless of the number of items}.  For example, consider an instance with two ``type A'' agents and a single resource.  Any feasible outcome (i.e., some way of splitting the resource into two pieces) corresponds to a point in two-dimensional space whose two coordinates are the utility of the first agent and the utility of the second agent, respectively, as demonstrated in Figure \ref{fig:Pareto_frontier}. 

As Figure \ref{fig:Pareto_frontier} illustrates, the set of utility points induced by feasible deterministic outcomes need not be convex and the shape depends on the actual valuation functions of the agents.  This is in contrast to the cake-cutting model, where the utility points from feasible outcomes form a convex set \cite{DS1961}. This non-convexity of the set of utility vectors from deterministic outcomes provides some intuition on why adding randomness enriches the space of feasible (expected) utilities in our model, and does not do so in the cake-cutting model.  This is because the set of all lotteries randomizes between the feasible outcomes can therefore yield any utility vector in the \emph{convex hull} of the utility points achievable by deterministic outcomes.  In short, randomness allows us to ``convexify'' the outcome space (as the dashed line does in Figure \ref{fig:Pareto_frontier}).  Note that all ex-ante Pareto efficient lotteries lie on a face of an (up to) $n$-dimensional convex hull. As a result, we can conclude that all utility points on the ex-ante Pareto frontier can be achieved with a lottery of support size at most $n$.

\begin{figure}[H]
\centering
\includegraphics[width=0.6\columnwidth,trim=0.1in 3.5in 0.1in 3.5in]{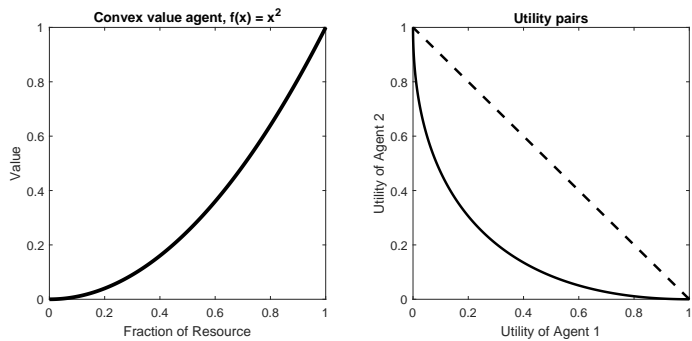}
\caption{At right is a solid line representing the pairs of utilities of all feasible outcomes when a good is completely allocated among two agents each with the valuation pictured at left.  The dashed line shows the pairs of expected utilities by randomizing between the two endpoint  outcomes.}
\label{fig:Pareto_frontier}
\end{figure}

Things become more complicated when combining an agent with a convex valuation function and one with a concave one, as in Figure \ref{fig:EF_frontier}. The graph at the bottom of this figure demonstrates the expected utility pairs that can be achieved for these two agents using ex-ante envy-free lotteries. It is worth noting that the ``all-or-nothing'' lottery (which gives each agent the whole item with probability $0.5$ and yields expected utility $0.5$ for both of these agents) is clearly Pareto dominated by other ex-ante envy-free lotteries. The points in red correspond to the utility pairs induced by the set of ex-ante Pareto optimal EF lotteries.

\begin{figure}[H]
\centering
\includegraphics[width=0.9\columnwidth,keepaspectratio,trim=1.2in 2.5in 1.2in 2.5in]{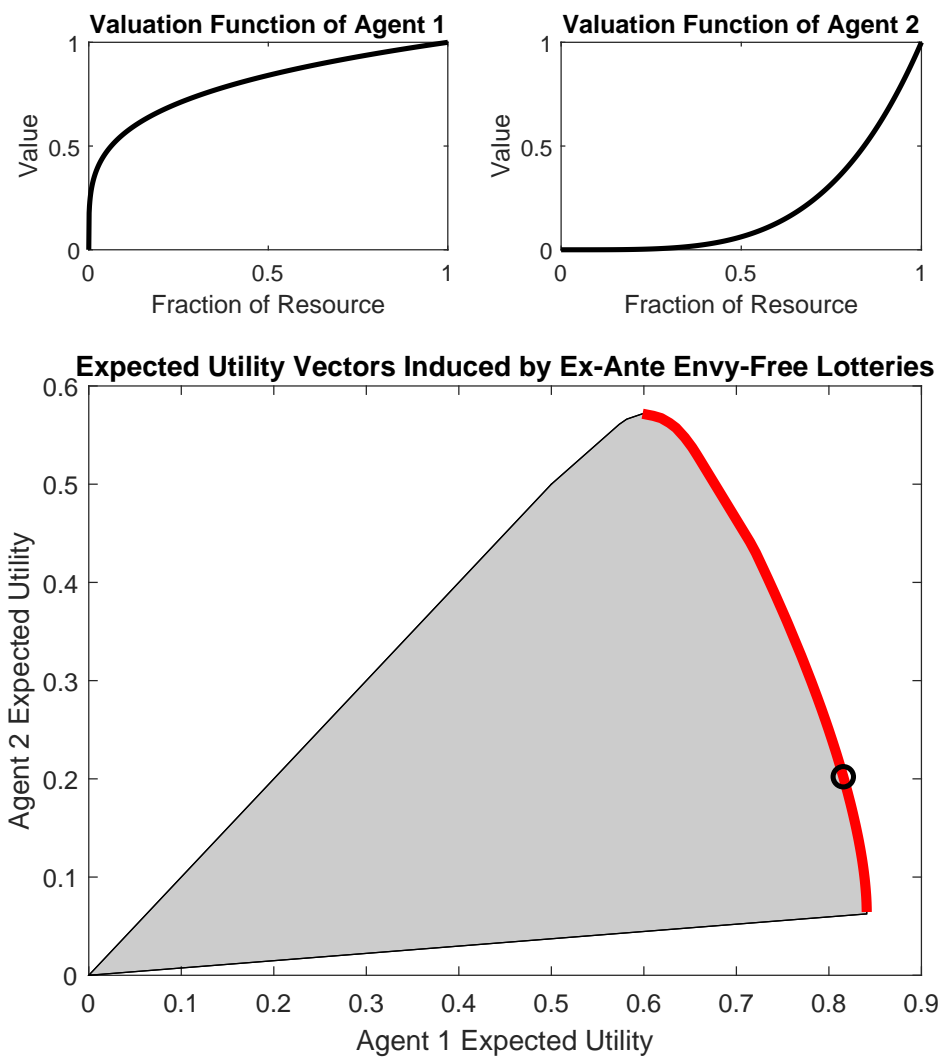}
\caption{At top left is the concave valuation function of Agent 1.  At top right is the  convex valuation function of Agent 2. At the bottom, the shaded region contains all the pairs of utilities (for Agent 1 and Agent 2) induced by ex-ante EF lotteries. The red bolded faces form the ex-ante Pareto optimal EF lotteries. For example, the circled point corresponds to the lottery where Agent 1 is allocated approximately $0.75$ of the good with probability $0.7$ and $0.1$ of the good with the remaining $0.3$ probability (the allocation of Agent 2 is the complement).  
}
\label{fig:EF_frontier}
\end{figure}

The situation becomes even less clear when agents are neither strictly concave nor convex.
For example, how should an item be allocated when some agents have ``type C'' valuations? Furthermore, things get more complicated with multiple resources.
For instance, when all agents have ``type A'' valuations for all resources, i.e., are risk-seeking and prefer lotteries, but have different ``favorite'' resources, allocating each agent her favorite resource with probability $1$ could Pareto dominate a uniform lottery over the entire bundle of goods.  
The space of outcomes and lotteries is so vast that it can be difficult to even guess what an ex-ante Pareto efficient EF lottery looks like.

\section{Query complexity}\label{sec:lowerbound}

In this section, we study the question of query complexity in our model. Specifically, we are interested in the number of queries necessary for finding an approximately ex-ante Pareto optimal EF lottery. Our main result is a lower bound.

\begin{theorem}\label{thm: lower bound}
Computing an ex-ante $\frac{\epsilon}{16}$-Pareto optimal EF lottery with any deterministic adaptive protocol requires at least $\frac{1}{2\epsilon}$ queries, even just for two items and two agents.
\end{theorem}

We briefly sketch the main intuition of the proof. Given an algorithm that terminates after $k$ queries we simulate its behavior when the responses to all queries are as if the underlying valuation functions were linear. For every item $i$, there must be an interval $[x_1,x_2]$, of size at least $1/k$, such that the algorithm only knows $f_i(x_1)$ and $f_i(x_2)$ (which are $x_1$ and $x_2$, respectively, since we were responding as if everything was linear). That is, in the true underlying instance, $f_i$ could be linear in $(x_1,x_2)$, or could be equal to $x_1$ in $(x_1,x_2 - \delta)$ and then increase more rapidly to $x_2$, or any other non-decreasing function. The crux of the argument is that there is no lottery that is simultaneously $\epsilon/16$-Pareto optimal among all EF lotteries for all different possible instances (three different possibilities suffice).

\begin{proof}[Proof of Theorem~\ref{thm: lower bound}]
Consider two agents $\{1,2\}$ and two goods $\{a,b\}$. Fix any deterministic (but possibly adaptive) algorithm asking at most $1/(2\epsilon)$ queries.
We respond to all queries as if all valuation functions were linear. After $1/(2\epsilon)$ queries the algorithm terminates and outputs some lottery $L$.

We construct three instances, $\mathcal{I}_1$, $\mathcal{I}_2$ and $\mathcal{I}_3$, consistent with the query responses. Intuitively, in these instances good $a$ will act as a ``currency'' which the two agents can exchange but good $b$ will have complementarities ``hidden'' between consecutive, sufficiently-spaced queries.  Notice that, by the pigeonhole principle, for item $b$, there must be an interval $[x_1,x_2]$ of size at least $\epsilon$, where $x_1 > 1/2$, such that the algorithm only ``knows'' that $f_{2b}(x_1) = x_1$ and $f_{2b}(x_2) = x_2$. Formally, the algorithm has never asked a value query with an endpoint in $(x_1,x_2)$, or a cut query for a value in $(x_1,x_2)$. Our three instances defer only inside this interval.

In $\mathcal{I}_1$ the valuation functions are linear everywhere.
In the second instance, $\mathcal{I}_2$, agent 1 has, again, linear valuation functions for both goods, and agent 2 has linear valuation  for  good $a$.
But, for good $b$ we have that $f_{2b}(x) = x_1 + 2(x - x_1)$ for all $x \in \left[x_1, x_1 + \epsilon/2\right)$, and $f_{2b}(x) = x_1 + \epsilon$ for all $x \in \left[x_1 + \epsilon/2, x_1 + \epsilon \right)$.  In other words, for the first half of the interval $[x_1, x_1 + \epsilon]$ the slope is $2$ and for the second half it is $0$. $f_{2b}$ is linear otherwise.  
The third instance, $\mathcal{I}_3$, is symmetric to $\mathcal{I}_2$ but with agents 1 and 2 flipped.

For $\mathcal{I}_1$, we note that any envy-free lottery must give both agents an expected total amount of the resources equal to $1$   
However, the algorithm cannot distinguish between $\mathcal{I}_1$, $\mathcal{I}_2$, and $\mathcal{I}_3$, so $L$ \emph{must} allocate an expected total amount of resources equal to $1$ to both agents.

Observe that in instances $\mathcal{I}_3$ and $\mathcal{I}_2$, the maximum utility that either agent can achieve when receiving a sum of fractions of the resources at most $1$ is $1 + \frac{\epsilon}{2}$.
This ex-ante utility is obtained  for agent 1 when she receives exactly $x_1 + \epsilon/2$ of good $b$ and $1  - (x_1 + \epsilon/2)$ of good $a$, and for agent 2 at her corresponding fractions in terms of $x_1$. But, these two outcomes cannot be achieved simultaneously since $x_1 > 0.5$.

Now, since both agents receive an expected total amount of resources equal to $1$, in $\mathcal{I}_2$, $L$ gives agent $1$ utility equal to $1$ (recall that in $\mathcal{I}_2$ agent $1$ is linear everywhere). But, there is an ex-ante envy-free lottery $L'$ for $\mathcal{I}_2$ with $u_1(L') \geq 1 + \frac{\epsilon}{8}$ and $u_2(L') \geq 1 + \frac{3\epsilon}{8}$: give to agent $2$ a $x_1 + \frac{\epsilon}{2}$ fraction of resource $b$ and a $(1 - x_1 - \frac{5}{8}\epsilon)$ fraction of resource $a$ and the rest to the agent $1$.
Since $L$ is ex-ante $\frac{\epsilon}{16}$-Pareto efficient for $\mathcal{I}_2$, it must give agent 2 utility strictly more than $1 + \frac{\epsilon}{4}$.  Similarly, $L$ must give agent 1 utility more than $1 + \frac{\epsilon}{4}$ in $\mathcal{I}_3$.

It remains to show that it is impossible for $L$ to allocate to each agent total resources equal to $1$, and simultaneously yield utility $1 + \frac{\epsilon}{4}$ in $\mathcal{I}_3$ for agent $1$ and $1 + \frac{\epsilon}{4}$ in $\mathcal{I}_2$ for agent $2$. Notice that, in each of $\mathcal{I}_2$ and $\mathcal{I}_3$, the expected utility of the non-linear agent is maximized (subject to the allocation constraint) when randomizing between the two outcomes/allocations that give the non-linear agent utility  $1 + \frac{\epsilon}{2}$ (the linear agent is always indifferent). But, no matter how this randomization is done, one of the two expected utilities (either $u_1$ in $\mathcal{I}_3$ or $u_2$ in $\mathcal{I}_2$) is weakly less than $1 + \frac{\epsilon}{4}$.
\end{proof}

\section{Approximately Optimal EF Lotteries}\label{sec:upperbound}
In this section, we develop the main positive result of the paper.  We complement our lower bound from the previous section by giving an algorithm which computes an ex-ante $\epsilon$-Pareto optimal EF lottery for any number of agents and items using only a polynomial number of queries.  Assume, for now, that there is a single item. We start by discretizing the item: ask each agent $i$ a series of $1/\epsilon$ value queries\footnote{For simplicity, we assume that $1/\epsilon$ is integer throughout this section.} $\textsc{Value}(f_i, x)$ for each $x \in \{\epsilon, 2\epsilon, 3\epsilon, \dots, 1\}$.  We do this with the aim of treating $1/\epsilon$ pieces of the resource (of size $\epsilon$ each) as atomic.  We then seek to find a lottery over the possible allocations of the $m/\epsilon$ ``indivisible'' pieces.

\subsection{A failed first attempt}
We define a variable $p_{i,y}$ denoting the probability that agent $i$ is allocated exactly $y$ pieces of the resource for $y \in [1/\epsilon]$ and write the linear program below.

\begin{talign*}
&\text{maximize}  \sum_{i=1}^{n} \sum_{y = 0}^{1/\epsilon} p_{i,y}f_i(y \cdot \epsilon) \\
&\text{subject to}  \sum_{i = 1}^{n} \sum_{y = 0}^{1/\epsilon}p_{i,y}\cdot y \cdot \epsilon = 1\\
& \quad\quad\quad\quad \sum_{y =0}^{1/\epsilon}p_{i,y} = 1, \quad \forall i \in \agents\\
& \sum_{y = 0}^{1/\epsilon}p_{i,y} f_i(y\cdot \epsilon)  \geq \sum_{y = 0}^{1/\epsilon}p_{i',y}f_i(y\cdot \epsilon), \forall i, i' \in \agents
\end{talign*}
The first set of constraints requires that the whole item is allocated. The second set restricts the variables $p_{i,y}$ so that they define a probability distribution for each agent $i$. The third set of constraints encodes the ex-ante envy-freeness requirements. Maximizing the total expected utility of the agents in the resulting lottery yields an ex-ante Pareto efficient EF lottery.

To arrive at an ex-post outcome, one could then allocate to each player $i$ a segment of size $y\epsilon$ with probability $p_{i,y}$.  Unfortunately, as the constraints of this linear program only guarantee that the entire resource is allocated in expectation, derandomizing is not guaranteed to produce a lottery over ex-post feasible outcomes.  For example, consider the following two-agent instance: agent 1 has value $1$ if she receives half of the item, and agent 2 has value $1$ if she receives the entire item.  
Allocating to the first agent half the item with probability $1$ and to the second agent the entire item with probability $1/2$ is feasible for this linear program (LP). This solution to the LP, of course, cannot be derandomized.

\subsection{Corrected attempt: A flow-based approach}
In this section we construct an alternative linear program which randomizes over explicitly feasible ex-post outcomes, thus correcting the issues of the previous LP.  We start by describing our approach for the single item case and then describe how to extend the approach to multiple items.

We aim to find a lottery over $1/\epsilon$ pieces, which we treat as indivisible.  This time, however, we select an arbitrary ordering over the $n$ agents and construct a directed graph with $n$ columns of $1/\epsilon+1$ vertices each. We add a directed edge from vertex $u$ in the $i$-th column and $j$-th row to each vertex $v$ in the $i+1$-th column and the $j+y$-th row for all $y \in \{0, 1, \dots, 1/\epsilon - j\}$. We also have a source and sink vertex connected to all vertices in the first and the last columns, respectively (see Figure~\ref{fig:flowGraph}). This construction guarantees that a path from the source to the sink vertex corresponds to a feasible allocation of the indivisible pieces of the good to the agents. The residual amounts of the good given to each agent in our arbitrary order is determined by the edges in the path as follows. Suppose the path contains an edge connecting vertex $u$ in the $i$-th column and $j$-th row to vertex $v$ in the $i+1$-th column and $k$-th row; this corresponds to allocating a fraction $\epsilon \cdot (k - j)$ of the item to agent $i+1$. An edge from the source to vertex $u$ in the first column and $j$-th row corresponds to allocating a fraction of $\epsilon \cdot j$ of the item to agent $1$. 

We can then assign to each edge from $i,j$ to $i+1, j+y$ a variable $p_{i,j,y} \in [0,1]$ corresponding to the probability that a $j\cdot \epsilon$ fraction was allocated up to (and including) agent $i$ and an additional $y\cdot \epsilon$ fraction was given to agent $i+1$, giving her a value of $f_{i+1}(y \cdot \epsilon)$.  A Pareto efficient lottery can then be found by maximizing a linear objective (e.g., the total value of all agents) subject to flow preservation constraints through this graph.  To ensure ex-ante envy-freeness, the desired fairness constraint can be encoded in the linear program. Overall, we find an approximately ex-ante Pareto optimal EF lottery which is a randomization over ex-post feasible outcomes (since every path is ex-post feasible). 

Although we have only described how to construct the flow graph for a single item, this flow graph approach is easily extensible to multiple items when agents have additive valuations across items.  One begins by constructing a separate flow graph for each item.  Then, by connecting the sink vertex of the $m'$-th item graph to the source vertex of the $m'+1$-th item graph for all $m' < m$, we can then obtain an approximately ex-ante Pareto efficient lottery for any number of agents and items in polynomial time by optimizing the linear program over the set of constraints implied by this ``chained'' flow graph.   

\begin{figure}
\centering
\begin{tikzpicture}[->,scale=0.75,transform shape]
\tikzstyle{flowvertex}=[circle, draw, minimum size=24pt]
\tikzset{my label/.style args={#1}{
  append after command={
    (\tikzlastnode.center) node {#1}
    }
  }
}

  \node[flowvertex, my label={s}] (s) at (0,3) {};
  
  \node[flowvertex, my label={$1, 0$}](A0) at (2,0) {};
  \node[flowvertex, my label={$1, \epsilon$}](A1) at (2,2) {};
  \node[flowvertex, my label={$1, 2\epsilon$}](A2) at (2,4) {};
  \node[flowvertex, my label={$1, 3\epsilon$}](A3) at (2,6) {};
    
  \node[flowvertex, my label={$2, 0$}](B0) at (4,0) {};
  \node[flowvertex, my label={$2, \epsilon$}](B1) at (4,2) {};
  \node[flowvertex, my label={$2, 2\epsilon$}](B2) at (4,4) {};
  \node[flowvertex, my label={$2, 3\epsilon$}](B3) at (4,6) {};
	
  \node[flowvertex, my label={$3, 0$}](C0) at (6,0) {};
  \node[flowvertex, my label={$3, \epsilon$}](C1) at (6,2) {};
  \node[flowvertex, my label={$3, 2\epsilon$}](C2) at (6,4) {};
  \node[flowvertex, my label={$3, 3\epsilon$}](C3) at (6,6) {};

  \node[flowvertex, my label={t}] (t) at (8,3) {};
  
    \path (s) edge (A0);
    \path (s) edge (A1);
    \path (s) edge (A2);
    \path (s) edge (A3);
	
	\path (A0) edge (B0);
	\path (A0) edge (B1);
	\path (A0) edge (B2);
	\path (A0) edge (B3);
	
	\path (A1) edge (B1);
	\path (A1) edge (B2);
	\path (A1) edge (B3);
	
	\path (A2) edge (B2);
	\path (A2) edge (B3);
	
	\path (A3) edge (B3);
	
	\path (B0) edge (C0);
	\path (B0) edge (C1);
	\path (B0) edge (C2);
	\path (B0) edge (C3);
	
	\path (B1) edge (C1);
	\path (B1) edge (C2);
	\path (B1) edge (C3);
	
	\path (B2) edge (C2);
	\path (B2) edge (C3);
	
	\path (B3) edge (C3);
	
	\path (C0) edge (t);
	\path (C1) edge (t);
	\path (C2) edge (t);
	\path (C3) edge (t);

\end{tikzpicture}
\caption{Example flow graph for $\epsilon = 1/3$ and three agents.  A vertex $v$ labeled ``$i, y\epsilon$'' indicates that $y$ pieces of the resource have been consumed along the path from $s$ to $v$.}
\label{fig:flowGraph}
\end{figure}
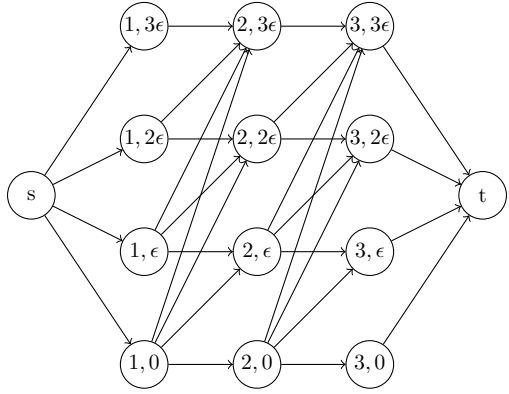

We present the linear program for a single item below.  
Let $\mathcal{P} = \{2,3,\dots,n\}\times\{0,1,2,\dots,1/\epsilon\}$.
The objective is again to maximize the total value of the agents. The first set of constraints impose flow preservation and the second constraint requires that a unit of flow is pushed from the source. The third set of constraints encodes ex-ante envy-freeness as agent $i$ must weakly prefer her own lottery to that of  $i'$: 

\begin{talign*}
&\text{maximize}  \sum_{i=1}^{n} \sum_{j = 0}^{1/\epsilon} \sum_{y = j}^{1/\epsilon} p_{i,j,y}f_i((y-j) \cdot \epsilon) \\ 
&\text{subject to} \sum\limits_{k = 0}^{1/\epsilon}{p_{i-1, k, j - k}} = \sum\limits_{\ell = 0}^{1/\epsilon} p_{i,j,\ell},  \forall (i,j) \in \mathcal{P}\\ 
&\quad\quad\quad\quad \sum_{j = 0}^{1/\epsilon}{p_{0,0,j}} = 1\\ 
&\sum\limits_{j = 0}^{1/\epsilon}{\sum\limits_{y = j}^{1/\epsilon}{p_{i,j,y}f_i( (y-j)\epsilon )}} \geq \sum\limits_{j = 0}^{1/\epsilon}{\sum\limits_{y = j}^{1/\epsilon}{p_{i',j,y}f_i( (y-j) \epsilon )}},\\
&\quad\quad\quad\quad\quad\quad\quad\quad\quad\quad\quad\quad\quad\quad\quad\quad\quad\quad\forall i,i' \in \agents
\end{talign*}

\begin{theorem}\label{thm:mainthm}
The linear program above outputs an ex-ante $\epsilon$-Pareto optimal EF lottery in polynomial time using $O\left(\frac{mn}{\epsilon^2}\right)$ queries for any $\epsilon < \frac{1}{mn}$.
\end{theorem}

\begin{proof}
First note that ex-ante envy-freeness is implied directly from the constraints of the program and that the lottery implied by the variables is a randomization over ex-post feasible outcomes by the discussion above.  Let $C$ denote the maximum Lipschitz constant across all pairs of agents and resources.  To demonstrate that the lottery is ex-ante $\epsilon$-Pareto optimal using at most $\frac{Cmn}{\epsilon^2}$ queries, first observe that the lottery which gives each agent all the items with probability $1/n$ is a feasible solution to our program and is consistent with our discretization.  On the other hand, as the objective function is optimizing the sum of ex-ante utilities (subject to the constraints), the linear program outputs a lottery which is ex-ante Pareto optimal with respect to the discretization.  Therefore, the linear program must give at least one agent utility greater than or equal to $1/n$. The only reason that the solution for our program would not be ex-ante Pareto efficient with respect to the full valuation functions would be if there were jumps in the valuation functions between our queries.  However, since all agents have Lipschitz valuation functions with Lipschitz constant at most $C$, in any ex-post outcome, if we were to give every agent an additional piece of size $\frac{\epsilon^2}{C}$ of every item, the utility of an agent would increase by at most $m\epsilon^2$. 
For an agent with utility $v\geq 1/n$, we have $v \cdot (1 + \epsilon) \geq v + \frac{\epsilon}{n} > v + m\epsilon^2$. So, there is no lottery that improves her utility by a factor $1 + \epsilon$, even with respect to the exact valuation functions.  Finally, since there are a polynomial number of variables and constraints the program runs in polynomial time.
\end{proof}

We also note that our algorithm for finding ex-ante approximate Pareto efficient EF lotteries can be extended to agents with \emph{non-additive}, bounded gradient valuation functions for a constant number of items.  The procedure is largely the same, except we ask $\frac{1}{\epsilon^m}$ value queries for all possible combinations of $\epsilon$ sized pieces, e.g., for two items $(0,0), (0, \epsilon), (0, 2\epsilon), \dots , (0, 1), (\epsilon, 0), (\epsilon, \epsilon), \dots, (1, 1)$.  We then appropriately modify the flow graph to include vertices representing all possible tuples for each agent and direct edges accordingly.  However, this procedure only works for a constant number of items as the number of queries then depends \emph{exponentially} on the number of items.

Since we discretize the resources into $1/\epsilon$ pieces each and treat these pieces as indivisible items, one could consider using the approach of \citet{BCKM2013} to find a lottery over ex-post feasible outcomes.  However, our valuation functions over the items can be more general than their approach can support since we allow complementarity within a single resource (e.g., by allowing for convex valuation functions).  On the other hand, their approach allows for more expressive constraint structures over the feasible allocations of goods to agents.  We believe that investigating the degree to which the two approaches can be combined is an interesting question for future work.

\section{Additional Observations and Future Work}\label{sec:future}
In this section, we provide some auxiliary results in our model and suggest several interesting related open questions.  An interesting first question to consider would be to tighten the gap between the upper and lower bounds on the number of queries required to compute an ex-ante $\epsilon$-Pareto efficient EF lottery for a constant number of players.  Do $o\left(\frac{1}{\epsilon^2}\right)$ queries suffice?  In the case of general $m$ and $n$, do $o\left(\frac{mn}{\epsilon^2}\right)$ suffice? 

While our lower bound suggests that exact ex-ante Pareto efficient lotteries are unattainable using a bounded number of queries, if we simplify our objective from ex-ante Pareto efficiency to ex-post Pareto efficiency then we may find an ex-ante envy-free lottery which is  \emph{exactly} ex-post Pareto efficient with respect to the set of all outcomes.  
We note that our algorithm closely resembles the well-known random serial dictatorship mechanism for one-sided matching markets of \citet{RSD}.

\begin{theorem}\label{thm:rsd}
There exists an algorithm that produces an ex-ante envy-free lottery that is exactly ex-post Pareto efficient among the set of all outcomes using a polynomial number of queries.
\end{theorem}
\begin{proof}
Consider an arbitrary set of $n$ agents $\agents$. We uniformly and randomly permute the agents and label them $1, 2, \dots, n$.  We ask agent $1$ a value query \textsc{Value}$(f_{1k}, 1)$ for each item $k \in \items$.  In doing so, we determine her value $v_{1k}$ for receiving each good $k$ entirely. We then ask her a cut query \textsc{Cut}$(f_{1k}, v_{1k})$ for each item $k \in \items$, thereby determining how much of each good she needs to receive to obtain her ``full value''. We continue by allocating to her as much of each item as the corresponding cut query returns.  For any remaining pieces of each item, we repeat this process for each successive agent in the ordering.  As valuation functions are additive across items, it is easy to verify that this is ex-post Pareto efficient; removing any amount of any of the goods from agent $i$ must decrease her value. This lottery is also ex-ante envy-free, since every agent has the same probability of appearing at any index in the random ordering, and at every point in the algorithm each consecutive agent receives a favorite bundle of goods among the remaining resources. Therefore, no agent envies the lottery of any other agent.
\end{proof}

On the other hand, while ex-post Pareto efficiency is easier to achieve than ex-ante Pareto efficiency, ex-post envy-freeness is a more stringent requirement than ex-ante envy-freeness.  For a lottery to be ex-post envy-free, it must be that all outcomes in its support are ex-post envy-free.  One may then ask the same questions we explore in this paper through the context of the class of ex-post envy-free lotteries.  For instance, can one obtain an ex-ante $\epsilon$-Pareto efficient lottery among the set of all ex-post envy-free lotteries using only a polynomial number of queries? 

One can also imagine asking a similar question of finding ex-ante Pareto optimal lotteries among the set of all \emph{ex-ante proportional} lotteries.  Ex-ante proportionality is another standard notion of fairness which requires all agents to receive expected utility at least $1/n$.  By changing the third set of constraints in our linear program in Section \ref{sec:upperbound} to encode proportionality, one directly obtains the same guarantee of Theorem \ref{thm:mainthm}, i.e., one obtains an approximately ex-ante Pareto optimal proportional lottery.  However, perhaps even more excitingly, one could change the welfare-maximization objective and instead maximize the \emph{leximin} utility by repeatedly maximizing the minimum value (as in, e.g., \citet{KPS2018}) and \emph{drop the fairness constraint altogether}.  Since the equiprobable lottery is in the space of feasible lotteries, the leximin solution will necessarily give all agents expected utility at least $1/n$.  Thus, the lottery would be approximately ex-ante Pareto optimal among \emph{all} lotteries and satisfy ex-ante proportionality!  

We view the question of determining if one can find a lottery which is approximately ex-ante Pareto optimal among all lotteries and which is also ex-ante envy-free to be a very interesting and challenging question.  In fact,  \citet{brams2012maxsum} point to the related question of finding a tractable algorithm to compute welfare-maximizing envy-free and Pareto efficient allocations (among the set of all allocations) in the standard cake-cutting model even in the seemingly quite simple case of known piecewise constant valuations as their ``most important, and presumably quite challenging, open problem''.  More broadly, investigating the extent to which different fairness notions change the shape of the ex-ante Pareto frontier is an interesting line of future study.

\section{Conclusion}
In this paper, we propose a new model 
for allocating homogeneous divisible goods and we show that it is simple in natural ways -- envy-freeness is easy to achieve -- yet highly complex in others -- many descriptive, non-trivial value functions can be captured.  
We provide an algorithm for finding approximately ex-ante Pareto optimal EF lotteries and complement this result with a lower bound on the query complexity of doing so.  Notably, while our algorithm uses only value queries, our lower bound applies to any algorithm using either cut and/or value queries.
We also provide an algorithm for finding exactly ex-post Pareto optimal EF lotteries which uses both cut and value queries.

Throughout this paper, we assume that valuation functions are Lipschitz (similarly to the assumptions made by \citet{RP1998} and \citet{CLPP2011} in their analogous works in the standard cake-cutting model). This assumption eliminates the possibility of really sharp jumps in the valuation functions (such jumps would significantly complicate our ability to identify envy free and Pareto efficient outcomes). 
Therefore, determining if approximate Pareto optimality can be achieved when the Lipschitz assumption is removed would be a  first step in extending our results.  
Perhaps most importantly, removing the assumption that agents are additive across resources is a major open line of further work.  We have a way of converting our approach for an arbitrary number of additive items to a computationally efficient method of finding  approximately ex-ante Pareto efficient EF lotteries for general valuations if there is a constant number of items.  On the other hand, extending this result to arbitrarily many items may require very different techniques.

\bibliographystyle{plainnat}
\bibliography{bibliography}

\end{document}